\documentclass[onecolumn]{IEEEtran}
\usepackage{amsmath,amsfonts,amssymb,euscript, graphicx, 
epsfig,
enumerate,float,afterpage, subfigure, ifthen}%

\usepackage{url}
\usepackage{hyperref}

\newtheorem{thm}{Theorem}

\newtheorem{lem}{Lemma}

\newcommand{\expect}[1]{\mathbb{E}\left[#1\right]}

\begin{document}

\title{Repeated Games, Optimal Channel Capture, and Open Problems for Slotted Multiple Access}

\author{Michael J. Neely \\ University of Southern California\\ \url{https://viterbi-web.usc.edu/~mjneely/}
\thanks{Michael J. Neely is a Professor in the Electrical Engineering Department at the University of Southern California. This work was supported by grant NSF SpecEES 1824418.} 
}


\maketitle


\begin{abstract} 
This paper revisits a classical problem
of slotted multiple access with success, idle, and collision events
on each slot.   First, results of a 2-user multiple access game
are reported. The game was conducted at the University of Southern
California over multiple semesters and involved competitions 
between student-designed
algorithms. 
An algorithm called 4-State was a consistent winner.  This algorithm
is analyzed and shown to have an optimal expected score when
competing against an independent version of itself.  The structure
of 4-State motivates exploration of the open question of how
to minimize the expected time to capture the channel for a 
$n$-user situation.  It is assumed that the system delivers perfect
feedback on the number of users who transmitted at the end of each slot. 
An efficient algorithm is developed and conjectured to have an optimal
expected capture time for all positive integers $n$. Optimality is proven in 
the special cases 
$n \in \{1, 2, 3, 4, 6\}$ using a novel analytical technique that introduces 
virtual users with enhanced capabilities. 
\end{abstract} 

\section{Introduction} 

This paper studies simple uncoordinated schemes for human users to 
share a multiple access channel.  Such schemes are useful, for example, 
in internet-of-things situations where distributed users send bursts of data 
and must learn an efficient channel sharing rule based on feedback.  
In this direction, the first thrust of this paper considers a series of 2-user 
slotted multiple access competitions that were conducted at the 
University of Southern California (USC) over multiple semesters. Students were asked to develop their own algorithms 
for choosing to transmit, or not transmit, over 100 slots. All pairs of algorithms competed and the one with the best accumulated score was declared the winner.   Two particular algorithms, called 3-State and 4-State, consistently outperformed the others. 
Analytical properties of these winning algorithms are established. In particular, it is shown that they have a maximum expected score when playing against independent versions of themselves. Properties of these algorithms motivate the second thrust 
of this paper: Investigating the open question of how to minimize the expected time to capture the channel for a $n$-user situation.  This question is of fundamental interest because it explores the learning times and algorithmic protocols required for a collection of indistinguishable users to identify a single user via distributed means.  

Algorithms that minimize the first capture time can be used in a variety of contexts. For example, to maximize throughput, 
an extended algorithm might give the first-capturer unhindered access to the channel for $k$ additional slots (this amortizes the slots that were spent trying to establish the first success).  Alternatively, to provide fairness, 
an algorithm that minimizes the time to the first success given an initial collection of $n$ users might be 
recursively repeated for $n-1$ users, then $n-2$ users, and so on, in order to construct a fair transmission schedule (such as a 
round-robin schedule). 

The paper treats a classical multiple access scenario with success, collision, or idle on every slot.  However, 
for the case with more than two users, it is additionally assumed that the receiver gives feedback on the number of 
transmitters at the end of every slot.  This is more detailed feedback than collision, idle, or success.  The number of 
transmitters can be determined by various physical estimation techniques, such as measuring the combined energy 
in the collisions and/or by using the bit signature technique of \cite{zigzag-katabi} to count the number of signature sequences that
are received. Even with this detailed feedback, the problem of minimizing the expected time to capture the channel in a $n$-user situation  involves a complex and seemingly intractable
decision structure. The paper develops an algorithm that can be implemented for any number of users $n$. The algorithm is 
conjectured to be optimal for all $n$.  A proof of this conjecture for the special cases $n \in \{1, 2, 3, 4, 6\}$ is given.  Remarkably, 
the expected capture time for the case of 3 users is strictly smaller than the expected capture time for the case of 2 users.  In this context, it is better to have 3 users independently competing for access than to have 2 users independently competing for access. 
The proof uses a novel technique that introduces virtual users with enhanced capabilities. This identifies two open questions that we leave to future work: (i) Is the algorithm of this paper optimal for all values of $n$?  (ii)  What algorithms are optimal when more limited feedback information is used? 

This paper focuses on the case of a single channel.  
However, it should be noted that modern wireless  
networks use a broad frequency spectrum that can support multiple channels over orthogonal frequency bands. 
One way to extend the transmission schemes of this paper to a multi-channel environment is to 
implement them separately and independently over each channel.  Is such separate and independent implementation optimal, or can correlating the transmission decisions over each channel provide efficiency gains?  This question is briefly explored at the end of the paper where it is shown that channel independence is optimal for 2 users but suboptimal for 3 users.

\subsection{Decentralized control} 

The problem of distributed minimization of  expected capture time  is conceptually similar to the  
stochastic decision problem for multiple distributed agents
described in  \cite{open-problems-comms}. The work \cite{open-problems-comms}
provides an optimal strategy for $n=3$ agents but leaves the case $n>3$ open.
 To date, there are no known polynomial time solutions for general $n$. 
Optimal distributed strategies for extended problems are developed in \cite{dist-opt-ton},
although complexity can grow exponentially in the number of agents. 
Unlike \cite{open-problems-comms}\cite{dist-opt-ton}, 
the agents in the current paper are not labeled, so the first agent cannot distinguish itself as ``wireless user 1'' and the second cannot distinguish itself as ``wireless user 2.'' Further, the problem of minimizing the expected capture time  that is treated in this paper (seemingly) suffers 
from an even worse complexity explosion 
than the problems in \cite{open-problems-comms}\cite{dist-opt-ton}.   Indeed,  while \cite{dist-opt-ton} establishes 
an optimal decision rule with complexity that grows with $n$, the optimal decision structure for the
current paper is unknown and the author is unable to establish optimality even for the case $n=5$. An optimality argument is given for  
 $n \in \{1, 2, 3, 4, 6\}$ using a novel  technique that introduces virtual users.  
 
 The problem of this paper is similar to 
 the class of \emph{sequential team problems} treated in \cite{ashutosh-common-knowledge} 
 using concepts of \emph{common information} (see also 
 standard form representations and discussions of nonclassical information patterns in \cite{standard-form-nonclassical-info}). 
 In the current paper, common information arises from the feedback that is commonly given to all users.  It may prove fruitful 
 to cast the current paper in the framework of \cite{ashutosh-common-knowledge}\cite{standard-form-nonclassical-info}. However, 
 that framework  
 does not necessarily provide low complexity solutions. Further, the decision structure under that framework 
 is very different from the simple policy structure considered and conjectured to be optimal in the current paper. 
 An explicit connection to \cite{ashutosh-common-knowledge}\cite{standard-form-nonclassical-info} is left for future work. 

The works \cite{pmlr-v125-bubeck20a}\cite{pmlr-v125-bubeck20c}\cite{6763073}
treat distributed channel selection for multi-access problems 
using a multi-armed bandit 
framework and using the criterion of asymptotic regret.  Like the current paper, 
\cite{pmlr-v125-bubeck20a}\cite{pmlr-v125-bubeck20c}\cite{6763073} assume
the users are not labeled, cannot distinguish themselves, and experience collisions
if multiple users pick the same channel. The work 
\cite{pmlr-v125-bubeck20a} treats two users and three stochastic channels, 
\cite{6763073} treats multiple users, 
\cite{pmlr-v125-bubeck20c} treats multiple users in a non-stochastic setting.
The channel model of the current paper is simpler than
\cite{pmlr-v125-bubeck20a}\cite{pmlr-v125-bubeck20c}\cite{6763073} and can easily be shown to yield a constant regret on expected throughput. 
However, rather than seeking to minimize regret in an asymptotic sense, 
this paper seeks a more stringent form of optimality:   Maximizing throughput
over a finite horizon, and minimizing expected capture time over an infinite horizon. 
The resulting control decision structure is different from \cite{pmlr-v125-bubeck20a}\cite{pmlr-v125-bubeck20c}\cite{6763073} and gives rise to a number of open questions that we partially resolve in this work.

\subsection{Distributed MAC and repeated games} 

For randomly arriving data, the classical slotted Aloha protocol is well known to achieve stability with throughput close to $1/e$. 
Splitting and tree-based algorithms that are optimized for Poisson arrivals are treated in \cite{bertsekas-data-nets}\cite{MoH85}\cite{TsM80}\cite{Hay76}\cite{Cap77}\cite{TsM78} and shown to increase throughput. Algorithms of this type that achieve throughput of 0.4878
are developed in \cite{MoH85}\cite{TsM80}; The maximum stable throughput under these assumptions is unknown but an upper bound of 0.587 is developed
in \cite{MiT81}.   The current paper treats
a fixed number of users with an infinite number of packets to send, rather than randomly arriving users.  Thus, there is no stability concern and the corresponding distributed implementation issues are different. 

The multi-access game treated in thrust 1 of this paper is a \emph{repeated game}
and is 
inspired by the repeated prisoner dilemma games  
in \cite{axelrod-prisoner80}\cite{axelrod-prisoner81} (see also, for example, \cite{prisoner-dilemma-book}\cite{game-theory-book}\cite{AGT}).    It has been observed that 
competitions involving 
repeated prisoner dilemma games are often 
won by the simple \emph{Tit-for-Tat} strategy that mirrors 
the opponent decision \cite{axelrod-prisoner80}\cite{axelrod-prisoner81}. 
This is not the case for the
multi-access game treated in the current paper.  
While a Tit-for-Tat strategy can be used in the multi-access 
competitions, and in several semesters students submitted such 
algorithms, these algorithms did not win because: (i) Tit-for-Tat is deterministic and so it  
necessarily scores zero points when competing 
against itself; (ii) Tit-for-Tat performs poorly when it competes against 
an algorithm that never transmits.  An algorithm called 4-State consistently wins the competitions.   This algorithm has an initial randomization phase to capture the channel. It also has a punishing mechanism that seeks to drive the opposing algorithm to fairly take turns.  It is shown that 4-State achieves an optimal expected score when competing against an independent version of itself.

\subsection{Feedback details and physical layer} \label{section:feedback}

This paper assumes slotted time with  fixed length packets. Let $F[t]$ be the number of users who transmit on slot $t \in \{1, 2, 3, \ldots\}$.  A success occurs if and only if $F[t]=1$.  For the $n$-user situation it holds that $F[t] \in \{0, 1, 2, \ldots, n\}$ and 
\begin{itemize} 
\item $F[t]=0 \iff $ idle. 
\item $F[t]=1 \iff $ success.
\item $F[t] \geq 2 \iff $ collision. 
\end{itemize} 
The value of $F[t]$ is assumed to be 
given as feedback at the end of slot $t$. If $n=2$ then $F[t]$ is equivalent to the idle, success, collision feedback
of classical slotted
Aloha. If $n>2$ then the $F[t]$ feedback is more detailed. It is assumed that $F[t]$ can be inferred by the receiver even in the case of a collision. This can be done, for example, by measuring the energy in the combined interfering signals and assuming that this energy is proportional to $F[t]$.  Alternatively, the value of $F[t]$ can be inferred by installing a short signature bit pattern in every packet transmission and using a filter to count the number of patterns that arise. This method for counting the number of transmitters is used in the 
ZigZag multiple access scheme of \cite{zigzag-katabi}, which uses timing misalignments to perform interference stripping. A soft decision decoding version called SigSag is treated in \cite{sigsag-journal}. 
Like \cite{zigzag-katabi}\cite{sigsag-journal},  the current paper 
assumes $F[t]$ can be accurately counted. However, it does not 
consider interference decoding and treats two or more transmissions on the same slot as a collision from which no information is obtained (other than the number of packets that collided). It is worth noting that any realistic scheme for reporting $F[t]$ will have some probability of reporting error.  
Nevertheless, for simplicity,  it is assumed that $F[t]$ is reported without error. 

\subsection{Outline}

Sections \ref{section:game} and \ref{section:game-analysis} 
describe the 2-player repeated game and establish  
optimality  of 4-State and 3-State. 
Section \ref{section:multi-user} considers 
the channel capture problem for $n\geq 2$ users and develops an 
algorithm that is conjectured to be optimal for all  $n$. Section \ref{section:converse} 
proves a matching converse for   $n \in \{1, 2, 3, 4, 6\}$.  Section \ref{section:extension}   considers extensions to multi-channel problems. 

\section{2-player MAC game} \label{section:game} 

Consider the following 2-player multiple access (MAC) game: Fix $T$ as a positive integer.  The game
lasts over $T$ consecutive slots.  Two players compete to send fixed-length packets over a single channel during this time.  A single packet transmission takes one time slot. On each slot $t \in \{1, 2, \ldots, T\}$, each player makes a decision about whether or 
not to send a packet.  For each player $i \in \{1, 2\}$ and each slot $t \in \{1, 2, \ldots, T\}$, define 
 $X_i[t]  \in \{0,1\}$ as the binary decision variable that is $1$ if player $i$ decides 
to send on slot $t$, and $0$ else.  There are three possible outcomes on slot $t$: 
\begin{itemize} 
\item Idle: Nobody transmits, so  $(X_1[t], X_2[t])=(0,0)$. 

\item Success: Exactly one player transmits, so $(X_1[t], X_2[t]) \in \{(0,1), (1,0)\}$. 
\item Collision: Both players transmit, so $(X_1[t], X_2[t]) = (1,1)$. 
\end{itemize} 

Player $i \in \{1,2\}$ scores a point on slot $t$ if and only if it is the only player to transmit on that slot.  Let $F[t] \in \{0,1,2\}$ denote the number of transmissions on slot $t$, which is
equivalent to idle/success/collision feedback.    
Let $(S_1, S_2)$ be the score of each player at the end of the game: 
\begin{align*}
S_1 &= \mbox{$\sum_{t=1}^T X_1[t](1-X_2[t])$}\\
S_2 &= \mbox{$\sum_{t=1}^T X_2[t](1-X_1[t])$}
\end{align*}

Players 1 and 2 
know the value of $T$ and the idle/success/collision structure of the game. 
However,  they cannot distinguish themselves as Player 1 or Player 2.    Therefore, even if both players have the desire to fairly share the slots, so that one player transmits
only on odd slots and the other transmits only on evens, there is no a-priori way to decide who takes the odds
and who takes the evens. Each player knows its own decision on slot $t$. From the feedback $F[t]$ it can infer the decision of its opponent. For each player and each slot $t\geq 2$ define $H_{self}[t]$ and $H_{opponent}[t]$ as the history of decisions up to but not including  slot $t$ as seen from the perspective of that player.  For example Player 1 has
\begin{align*}
H_{self}[t] &= (X_1[1], X_1[2] , \ldots, X_1[t-1])\\
H_{opponent}[t] &=(X_2[1], X_2[2], \ldots, X_2[t-1])
\end{align*}
while these vectors are swapped for Player 2.

\subsection{Random and deterministic algorithms} 

Algorithms are allowed to make any desired decisions based on the feedback, including 
randomized decisions.  A general algorithm can be mathematically  represented by a 
sequence of Borel measurable functions $f_1, f_2, \ldots, f_T$ such that
\begin{align*}
&f_1:[0,1)\rightarrow \{0,1\} \\
&f_t:[0,1)\times \{0,1\}^{t-1} \times \{0,1\}^{t-1}\rightarrow \{0,1\} \quad \forall t \in \{2, \ldots, T\}
\end{align*}
where the decisions $X[t]$ are given by 
\begin{align*}
X[1] &= f_1(U)\\
X[t] &=f_t(U, H_{self}[t], H_{opponent}[t])  \quad \forall t \in \{2, \ldots, T\}
\end{align*}
where $U$ is a randomization variable that is uniformly distributed over $[0,1)$. The randomness of $U$ can be used to facilitate randomized decisions.\footnote{A random variable $U \sim U[0,1)$ has binary expansion $U = \sum_{m=1}^{\infty} B_m 2^{-m}$ with $\{B_m\}_{m=1}^{\infty}$ i.i.d. equally likely bits that can be used to make sequences of randomized decisions.}   Players are assumed to use independent randomization variables. In particular,  if Players 1 and 2 implement independent versions of the same algorithm, they use the same functions $f_1, \ldots, f_T$ but they use independent random variables $U_1 \sim U[0,1)$ and $U_2 \sim U[0,1)$.

A \emph{deterministic algorithm} is one that contains no randomization calls.  Such an algorithm can be characterized by a sequence of 
functions $\{g_t\}_{t=1}^T$ that only use $H_{self}[t]$ and $H_{opponent}[t]$ 
as input (with no randomization variable)
\begin{align*}
&g_1 \in \{0,1\} \\
&g_t:\{0,1\}^{t-1} \times \{0,1\}^{t-1} \rightarrow \{0,1\}\quad \forall t \in \{2, \ldots, T\} 
\end{align*}
so  $X[1]=g_1$ is the deterministic  decision on slot $t=1$ and 
\begin{align} \label{eq:det-alg} 
X[t] &=g_t(H_{self}[t], H_{opponent}[t]) \quad \forall t \in \{2, \ldots, T\} 
\end{align}

\begin{lem}  A deterministic algorithm scores zero points against itself.
\end{lem} 
\begin{proof} 
Let $g_1, \ldots, g_T$ be the functions associated with a deterministic algorithm that is used by both Player 1 and Player 2.  On the first slot we have 
$X_1[1]= g_1$ and $X_2[1] = g_1$ so both players make the same decision
(resulting in either a collision or idle).  
For $t \in \{2, \ldots, T\}$ define $H_{1,self}[t]$, $H_{1,opponent}[t]$ and $H_{2,self}[t]$, $H_{2,opponent}[t]$ as the history kept by Players 1 and 2, respectively.  By definition 
\begin{align}
H_{2,self}[t] &= H_{1,opponent}[t] \quad \forall t \label{eq:self1}\\
H_{2,opponent}[t] &= H_{1,self}[t] \quad \forall t \label{eq:self2} 
\end{align}
Since $X_1[1]=X_2[1]$ the following 
holds for $t=2$: 
\begin{equation}\label{eq:induction-hypothesis} 
H_{1,self}[t] =H_{2,self}[t] 
\end{equation} 
We proceed by induction: Fix $t \in \{2, \ldots, T\}$ and assume 
\eqref{eq:induction-hypothesis} holds for slot $t$.  
We show no player scores on slot $t$ and that if $t<T$ then 
\eqref{eq:induction-hypothesis} holds for $t+1$.   On slot $t$ we have 
\begin{align*}
X_1[t] &= g_t(H_{1,self}[t], H_{1,opponent}[t])\\
X_2[t] &= g_t(H_{2,self}[t], H_{2,opponent}[t])\\
&= g_t(H_{1,self}[t], H_{1,opponent}[t])
\end{align*}
where the last equality holds by \eqref{eq:induction-hypothesis} and \eqref{eq:self1}-\eqref{eq:self2}. Thus, $X_1[t]=X_2[t]$ and no player scores on slot $t$.  If $t<T$ then 
\begin{align*}
H_{1,self}[t+1] &= (H_{1,self}[t]; X_1[t])\\
H_{2,self}[t+1] &= (H_{2,self}[t]; X_2[t])\\
&= (H_{1,self}[t], X_1[t])
\end{align*}
where the final equality uses \eqref{eq:induction-hypothesis} and the fact $X_1[t]=X_2[t]$. Thus, $H_{1,self}[t+1]=H_{2,self}[t+1]$ so   \eqref{eq:induction-hypothesis} holds for $t+1$. 
\end{proof} 

\subsection{Tournament structure} 

Competitions of these 2-player MAC games were conducted over 7 semesters amongst students in the EE 550 Data Networks class at the University of Southern California.    If there were $k$ algorithms competing in a given semester, then all algorithms $i \in \{1, \ldots, k\}$  were paired against all other algorithms $j \in \{1, \ldots, k\}$.  This included a pairing $(i,i)$ where algorithm $i$ plays against an independent version of itself.   For each algorithm pair $(i,j)$, the scores of 1000 independent simulations of 100-slot games were averaged to provide an estimate of the expected score $(\expect{S_i}, \expect{S_j})$ associated with algorithm $i$ playing algorithm $j$. The total score of an algorithm is the sum of its scores accumulated 
over all other algorithms that it played (including itself).  The algorithm with the largest accumulated score was declared the winner of the competition. 

The competing algorithms included algorithms that students designed, an instructor-designed algorithm called 4-state, and two special algorithms called AlwaysTransmit and NeverTransmit. 
AlwaysTransmit transmits on every slot regardless of history.  No opponent can score against AlwaysTransmit. This algorithm is maximally greedy and was entered into the competition in order to view its accumulated score in comparison with the other algorithms.  In contrast, NeverTransmit never transmits and never scores any points. It was entered into the competition to test the ability of other algorithms to adapt to the situation where they are the only ones requesting channel access.

In addition to the overall score in the competition, the quality of an algorithm can be understood in terms of figures of merit $\alpha$ and $\beta$ defined below: 

\begin{itemize} 
\item Self-competition score $\alpha$: This is the expected score $\expect{S_1}$ when an algorithm plays an independent copy of itself.  This figure of merit is  useful  because a good MAC algorithm will be used by others and hence must perform well against itself. This is also important in the tournament because each algorithm plays against itself. 

\item No-competition score $\beta$: This is the expected score $\expect{S_1}$ when an algorithm plays NeverTransmit. This figure of merit is useful because a good MAC algorithm should adapt when nobody else is using the channel.  It is also important in the tournament because algorithms can earn close to 100 points on games when they play NeverTransmit (provided they can quickly guess when they are playing NeverTransmit). 
\end{itemize} 

\subsection{Special algorithms} 

The following algorithms are of key interest: 

\begin{enumerate} 
\item Tit-for-Tat-0 (TFT-0): This is a deterministic policy that operates as follows:
\begin{itemize} 
\item $X[1]=0$.
\item For $t \in \{2, 3, \ldots, T\}$:  $X[t] = X_{opponent}[t-1]$. 
\end{itemize} 
A variation called Tit-for-Tat-1 (TFT-1) differs only by having $X[1]=1$.  Both Tit-for-Tat-0 and Tit-for-Tat-1 mirror the opponent decisions with one slot delay. They give a free slot immediately after the opponent gives a free slot.  This is similar in spirit to the (different) Tit-for-Tat algorithm considered for prisoner dilemma games \cite{axelrod-prisoner80}\cite{axelrod-prisoner81}.  It can be shown that neither version ever looses a game by more than 1 point (regardless of the opponent). However, both Tit-for-Tat-0 and Tit-for-Tat-1 are deterministic and so they both have a Self-Competition score of 0.  They also have poor No-Competition scores (0 and 1 for Tit-for-Tat-0 and Tit-for-Tat-1, respectively).

\item 3-State:  This algorithm is given by the diagram in Fig. \ref{fig:3-State}.  The initial state is state 1.  The idea behind 3-State is to start by randomly transmitting  until the first success, then move to a turn-based policy that oscillates between states 2 and 3.  On its turn (state 3) it repeatedly transmits until it scores.  This is a punishing mechanism designed to force the opponent to acknowledge its turn. Like Tit-for-Tat, it is not difficult to show   that 3-State cannot loose a game by more than 1 point. 

\item 4-State: This algorithm is given by the diagram in  Fig. \ref{fig:4-State}. The idea is to notice that 3-State wastes half of the slots when playing NeverTransmit. The 4-State algorithm is the same as 3-State, with the exception that it attempts to detect whether or not it is playing NeverTransmit.     If the opponent does not take its turn in state 2, the algorithm moves to state 4 and repeatedly transmits on all remaining slots (unless there is a collision). It is not difficult to show that 4-State retains the property that it cannot loose a game by more than 1 point. 
\end{enumerate}

\begin{figure}[htbp]
   \centering
   \includegraphics[width=3in]{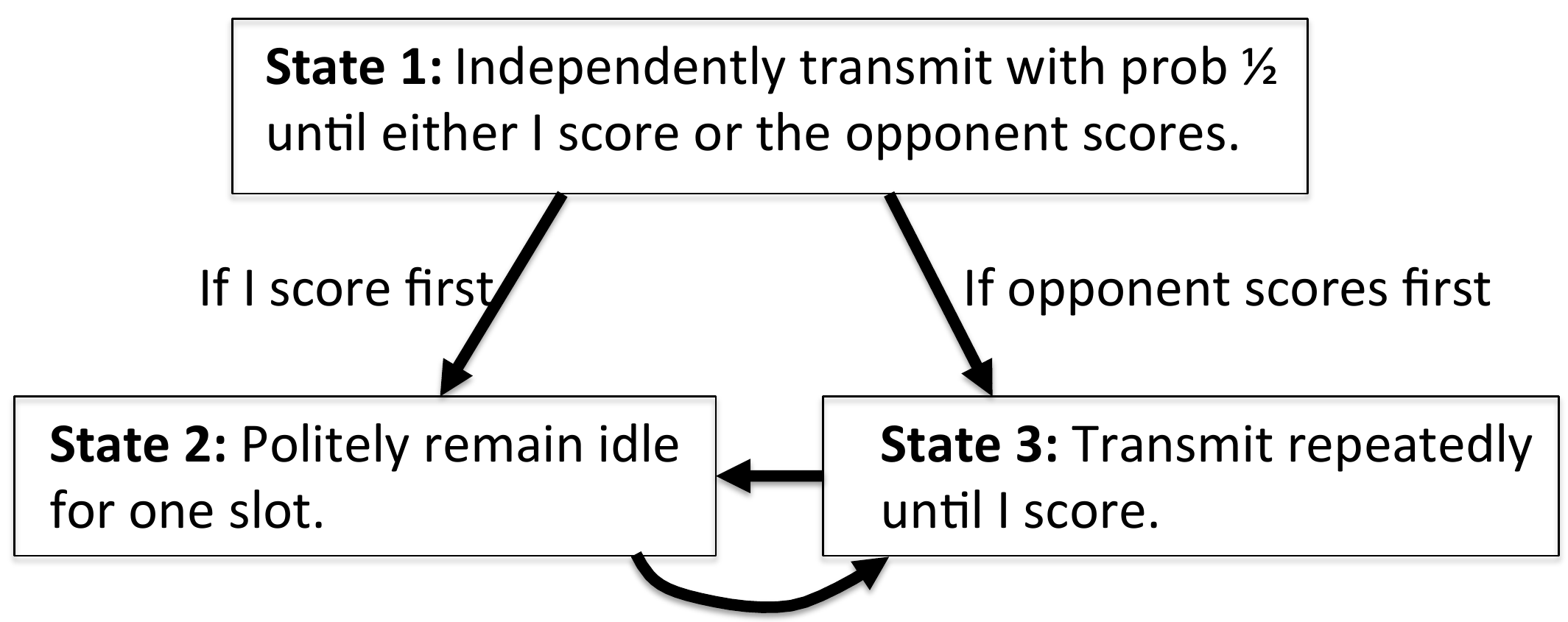} 
   \caption{Algorithm 3-State.  The algorithm starts in state 1.}
   \label{fig:3-State}
\end{figure}

\begin{figure}[htbp]
   \centering
   \includegraphics[width=3in]{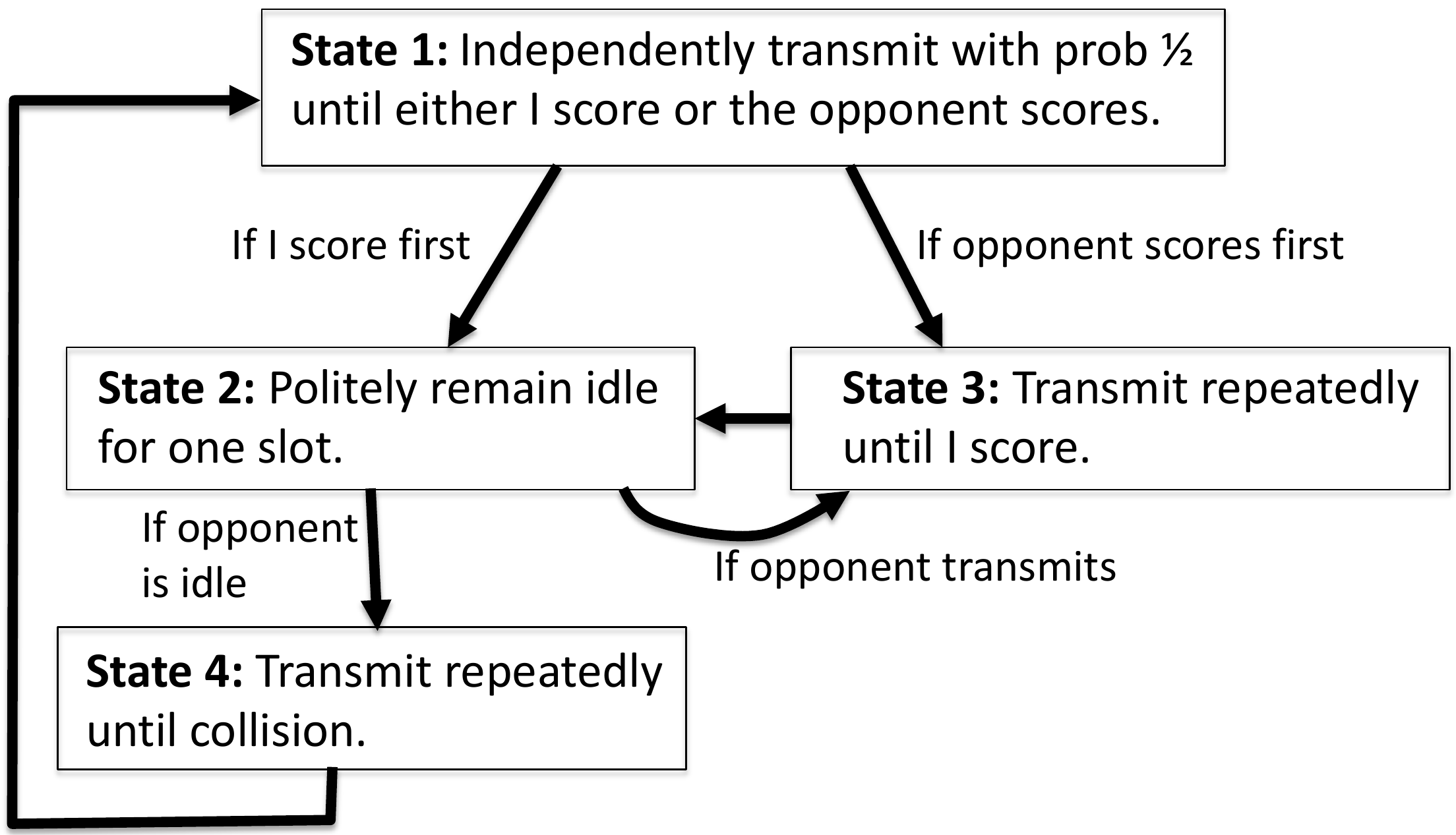} 
   \caption{Algorithm 4-State. The algorithm starts in state 1.}
   \label{fig:4-State}
\end{figure}

\subsection{Performance over multiple competitions} 

\begin{figure*}[htbp]
\centering
\begin{tabular}{|c||c|c|c|c|}
\hline
& 4-State & Second Place & AlwaysTransmit& AvgAlg  \\ \hline \hline
Fall 2021 (10 algs)& 32.46 & 26.02& 22.90 & 18.14 \\ \hline
Fall 2020 (25 algs)& 23.92 & 22.82 &12.36  & 12.10 \\ \hline
Fall 2019 (19 algs) & 30.55& 30.07  & 18.32& 16.25  \\ \hline
Spring 2018 (35 algs) & 56.31& 53.62& 25.55& 33.71 \\ \hline
Fall 2018 (27 algs) & 32.44 &29.63 & 15.42 & 17.11 \\ \hline
Spring 2017 (21 algs) & 20.44& 17.68 &8.00 & 10.88 \\ \hline
Fall 2016 (14 algs) & 20.22 & 17.53 & 11.22 & 10.22 \\ \hline
\end{tabular} 
\caption{Results for 7 semesters of competitions. The 4-State algorithm came in first place in each semester.  The second place score (which was from a student-designed algorithm) is also shown. AvAlg represents the average score over all algorithms for that semester.}
\label{fig:many-semesters} 
\end{figure*}

Fig. \ref{fig:many-semesters} provides results for 7 semesters of competitions.  Each competition 
included the 4-State, AlwaysTransmit, and NeverTransmit algorithms together with 
algorithms developed by students. 
 The number of competing algorithms varies semester by semester as shown in the table.  
 Every semester, all algorithms of the competition played against 
 all  others (including independent versions of themselves).   Each algorithm pair is simulated for 1000 independent runs of 100-slot games. Fig. \ref{fig:many-semesters}  records the average points per game.

 The 4-State algorithm came in first place on each semester.  Fig. \ref{fig:many-semesters}  also provides the average points per game of the second place algorithm, the AlwaysTransmit algorithm, and the average over all algorithms in that semester.  The 4-State algorithm consistently scores significantly more than AlwaysTransmit.  The average scores over all algorithms tended to be close to AlwaysTransmit because many students used algorithms that emphasized greedy decisions.

\subsection{Fall 2021 competition} 

\begin{figure*}[htbp]
   \centering
   \includegraphics[width=5.0in]{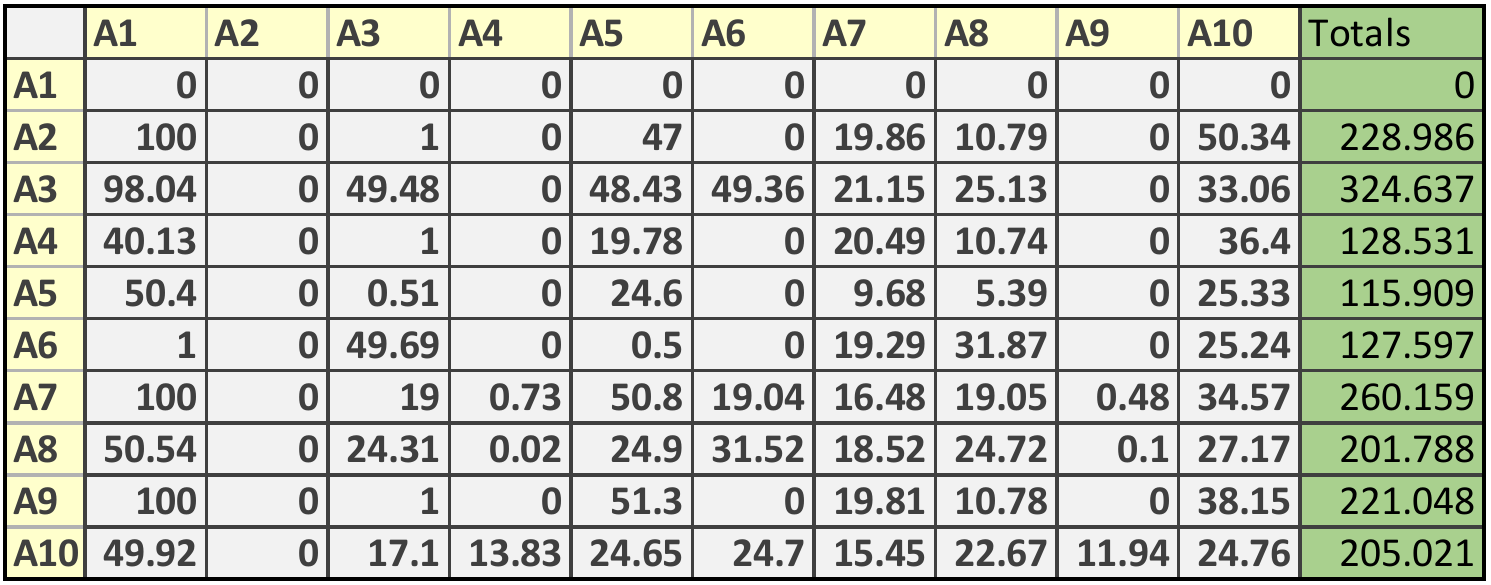} 
   \caption{Results of multi-access game in Fall 2021 semester. Each game consists of 100 slots. 
   Results are averaged over 1000 independent games.  The score of Algorithm $i \in \{1, ...,10\}$ is given in row $i$.} 
   \label{fig:table}
\end{figure*} 

A closeup look at the Fall 2021 competition is given in Fig. \ref{fig:table}.  That was a
small competition with 4-State, AlwaysTransmit, NeverTransmit, and only 7  student algorithms. 
 Each row $i \in \{1, \ldots, 10\}$ of 
Fig. \ref{fig:table} shows the score when algorithm $i$ played against each other algorithm (averaged over 1000 independent 100-slot games). Key algorithms of the 10 are: 
\begin{itemize} 
\item A1: NeverTransmit
\item A2: AlwaysTransmit
\item A3: 4-State
\item A6: This student used Tit-for-Tat-1
\end{itemize}

The 4-State algorithm (A3) dominated the competition by earning $324.637$ points.  Its simulated Self-Competition score of $49.48$ is shown in cell $(A3, A3)$ and this is consistent with the analytical result of the next section. The AlwaysTransmit algorithm (A2) earned $228.986$ points.  Tit-for-Tat-1  (A6) earned 
only $127.597$ points.  Its poor performance was due to a Self-Competition score of 0 and a No-Competition score of only 1.  However, it learns to fairly share with 4-State (earning $49.69$ points in that game, see the (A6, A3) cell). In tournaments with a larger number of students, there are more student algorithms that learn to share and Tit-for-Tat-1 performs better (as shown in Section \ref{section:merit}). 

The 4-State algorithm receives an average score per game of 32.46.  
It is interesting to note the following data which is not in the table: When 4-State is replaced with 3-State, the average score per game of 3-State is 24.68 (it comes as a close second place to the winning algorithm A7 which receives 25.80); when replaced by Tit-for-Tat-0 the average score per game for Tit-for-Tat-0 is 12.18; that for Tit-for-Tat-1 is 7.76.  Thus, the greedy version of Tit-for-Tat does worse.

\subsection{Transmitting on the last slot} 

Should an algorithm surely transmit on the last slot? Arguably, this is desirable because the opponent has no chance to retaliate.  However, this policy is harmful 
when an algorithm competes against an independent version of itself because 
both versions would necessarily receive 0 points on the last slot.  In view of this, an improvement on 4-State is the following: Implement 4-State, with the exception that if at any time  the opponent makes a decision that reveals it is \emph{not} 4-State (such as when the opponent fails to abide by the turn structure), then implement 4-State as usual on all but the last slot, and  transmit surely on the last slot.  When 4-State was replaced with this 
enhanced version of 4-State
in the Fall 2021 competition, it gained an extra 1.43 points in total (an average of 0.143 extra points per game).

\subsection{Figures of merit} \label{section:merit} 

The first two rows of Fig. \ref{fig:tournament} provide analytical values for the Self-Competition score $\alpha$ and the No-Competition score $\beta$  of  4-State, 3-State, Tit-for-Tat-0, Tit-for-Tat-1, and AlwaysTransmit (values 
of $\alpha$ and $\beta$ for 4-State and 3-State are derived in Section \ref{section:game-analysis}). The analytical values in the first two rows are consistent with simulation results. The 
third row in Fig. \ref{fig:tournament}  provides a ``Tournament Score'' $\gamma$ which is a simulation result on the average score per game, considering all games played, 
in a large tournament with 
135 algorithms (including these five algorithms together with all student algorithms that were created over 7 
semesters). Each algorithm was paired against all 135 other algorithms (including an independent version of itself) 
in 1000 independent runs of 100-slot games.   For this large tournament it can be seen that 3-State and 4-State are the best; the greedy version of Tit-for-Tat does significantly worse than the non-greedy version; the AlwaysTransmit algorithm does significantly worse than 3-State, 4-State, Tit-for-Tat-0, and Tit-for-Tat-1.

\begin{figure}[htbp] 
\centering
\begin{tabular}{|c||c|c|c|c|c|}
\hline
& 4-State & 3-State & TFT-0  & TFT-1 & AlwaysTran\\ \hline \hline
$\alpha$ & 49.500& 49.500& 0&  0& 0 \\ \hline
$\beta$  & 98.000& 49.667& 0& 1 &  100\\ \hline
$\gamma$ & 24.613& 22.548& 20.410& 15.326&10.714 \\ \hline
\end{tabular} 
\caption{A table of scores  (Self-Competition $\alpha$; No-Competition $\beta$; Tournament $\gamma$) for 4-State, 3-State, Tit-for-Tat-0, Tit-for-Tat-1, and AlwaysTransmit.}
\label{fig:tournament}  
\end{figure}

\section{Analysis of the 2-player MAC game} \label{section:game-analysis}

\subsection{Self-competition score of 4-State and 3-State}

\begin{thm} \label{thm:self-score} Fix $T$ as a positive integer.  In a $T$-slot game, the   
Self-Competition score for both 4-State and 3-State is: 
\begin{equation} \label{eq:self-score-3state} 
\alpha = \frac{T-1}{2} + \frac{1}{2^{T+1}}
\end{equation} 
Further, no other algorithm can achieve a larger Self-Competition score. 
\end{thm} 

\begin{proof} (Theorem \ref{thm:self-score} achievability)  Consider two independent versions of 4-State that compete over $T$ slots.  Let $S_1$ and $S_2$ be the resulting scores of the two players.   Since the algorithms are identical we have 
\begin{equation} \label{eq:same-score} 
\alpha = \expect{S_1} = \expect{S_2} 
\end{equation} 
Let $Y \in \{0, 1, \ldots, T\}$ denote the random number of 
initial slots in which nobody scores ($Y=T$ if nobody ever scores).  When 4-State plays against itself, once the first player scores, exactly one player will score on each slot thereafter.  Thus
$$ T = Y+S_1+S_2$$
Taking expectations of both sides and using \eqref{eq:same-score} gives
$$ T = \expect{Y} + 2\alpha$$
Thus
\begin{equation} \label{eq:alpha-pre} 
\alpha  = \frac{T - \expect{Y}}{2}
\end{equation} 
The random variable $Y$ has the following distribution
\begin{align}
P[Y=i] &= (1/2)^{i+1} \quad \forall i \in \{0, 1, \ldots, T-1\} \label{eq:Y1} \\
P[Y=T] &= (1/2)^T \label{eq:Y2}
\end{align}
Hence 
\begin{equation} \label{eq:Y}
\expect{Y}=   \sum_{i=0}^T i P[Y=i] = 1 - 2^{-T}
\end{equation} 
Substituting this expression for $\expect{Y}$ into \eqref{eq:alpha-pre} proves \eqref{eq:self-score-3state} for 4-State.  It can be shown that 4-State never uses state 4 when playing against itself, and hence $\alpha$ is also the Self-Competition score of 3-State. 
\end{proof} 

\begin{proof} (Theorem \ref{thm:self-score} converse)  Consider any algorithm that is independently used for both players.  For convenience, assume the algorithm is designed to run over an infinite sequence of slots $t \in \{1, 2, 3, \ldots\}$ (an algorithm that is designed to run over only a finite number $T$ of slots can be extended to run over an infinite horizon by choosing to never transmit after time $T$).    Arbitrarily assign one of the algorithms as Player 1 and assign its counterpart (identical) algorithm as Player 2.  For each positive integer $T$, define $V[T]$ as the random sum of scores of both players over the first $T$ slots: 
$$ V[T] = \sum_{t=1}^T \left[X_1[t](1-X_2[t]) + X_2[t](1-X_1[t])\right] $$
By symmetry, it follows that the Self-Competition score  (over $T$ slots)  is $\expect{V[T]}/2$.  For each positive integer $T$, define $s[T]$ as the supremum value of $\expect{V[T]}$ over all algorithms that independently compete against themselves.  Define $s[0]=0$. We want to show $s[T]/2 \leq \alpha$ for all nonnegative integers $T$, where $\alpha$ is the value in \eqref{eq:self-score-3state}.  Specifically, we want to show
\begin{equation} \label{eq:want-induction}
s[T] \leq T-1 + (1/2)^{T} \quad  \forall T \in \{0, 1, 2, 3, ...\} 
\end{equation} 
We use induction: Suppose \eqref{eq:want-induction} holds for $T=k$ for some nonnegative integer $k$ (it holds for $k=0$ since $s[0]=0$). We show it also holds for $T=k+1$. 
Since at most one player can score on each slot, we surely have 
\begin{equation} \label{eq:T-surely} 
V[k+1]\leq k+1
\end{equation} 

Let $A$ be the event that there is a success by one of the players on slot $1$, so that $A^c$ is the event that the first slot results in either an Idle or a Collision.  A key observation is
\begin{equation} \label{eq:key-observation} 
\expect{V[k+1]|A^c}\leq s[k]
\end{equation} 
since the event $A^c$ means that neither player scored on the first slot, there are $k$ slots remaining to accumulate the total score $V[k+1]$, and 
no information has been conveyed to either player on this first slot that would make the expected score over the remaining $k$ slots  larger than   $s[k]$.\footnote{If a non-success on the first slot made the expected score on the remaining slots \emph{more} than $s[k]$, one could use an algorithm that starts under the assumption that a non-existent preliminary slot just had a non-success. That would achieve an expected $k$-slot score that is larger than $s[k]$, a contradiction.}  

Let $q$ be the probability that the algorithm transmits on the very first slot.  Then
\begin{equation} \label{eq:PA-bound}
P[A] = 2q(1-q) \leq \sup_{q \in [0,1]} 2q(1-q) = 1/2 
\end{equation} 
We have 
\begin{align*}
&\expect{V[k+1]} \\
&=\expect{V[k+1]|A} P[A] + \expect{V[k+1]|A^c}(1-P[A]) \\
&\overset{(a)}{\leq} (k+1)P[A] + s[k](1-P[A])\\
&= s[k]+ P[A](k+1-s[k])\\
&\overset{(b)}{\leq} s[k] + (1/2)(k+1-s[k]) \\
&= (1/2)(k+1) + (1/2)s[k]\\
&\overset{(c)}{\leq}(1/2)(k+1)+(1/2)(k-1 + (1/2)^k)\\
&= k + (1/2)^{k+1}
\end{align*}
where (a) holds by \eqref{eq:T-surely} and \eqref{eq:key-observation}; (b) holds by \eqref{eq:PA-bound} and the fact $k+1-s[k]\geq 0$ (observe that $s[k]\leq k$ since at most one point can be scored per slot); (c) holds by the induction assumption 
that \eqref{eq:want-induction} holds for $T=k$.  Thus
$$ \expect{V[k+1]}\leq k + (1/2)^{k+1}$$
This holds for all algorithms. Taking the supremum value of $\expect{V[k+1]}$  
over all possible algorithms gives
$$ s[k+1] \leq k + (1/2)^{k+1}$$
which proves that \eqref{eq:want-induction} holds for $T=k+1$.
\end{proof} 

\subsection{No-competition score} 

\begin{lem} Fix $T$ as a positive integer.  In a $T$-slot game, the No-Competition scores $\beta_4$ and $\beta_3$ for 4-State and 3-State, respectively, are:  
\begin{align}
\beta_4 &= T-2 + \frac{3}{2^{T}} \label{eq:beta4}\\
\beta_3 &= \left\{\begin{array}{cc} 
  \frac{T}{2} - \frac{1}{3} + \frac{(1/3)}{2^{T}}& \mbox{if $T$ is even} \\
 \frac{T}{2} - \frac{1}{6} +\frac{(1/3)}{2^T}& \mbox{if $T$ is odd} 
 \end{array}\right. \label{eq:beta3} 
\end{align} 
\end{lem} 

\begin{proof} 
To compute $\beta_4$, suppose 4-State plays NeverTransmit and 
let $Y \in \{0, 1, 2, \ldots, T\}$ be the number of initial slots in which nobody scores. 
This $Y$ has the same distribution as \eqref{eq:Y1}-\eqref{eq:Y2}, and so by \eqref{eq:Y} we have $\expect{Y}=1-2^{-T}$. 
Let $S$ be the total score.  Then 
$$ S= T-Y - 1 + 1_{\{Y \in \{T-1, T\}\}}$$
which holds because if $Y \notin \{T-1, T\}$ then there are a total of $Y+1$ slots in which nobody scores (including the single idle slot given to the opponent when 4-State moves to state 2), and only $Y$ such slots if $Y \in \{T-1, T\}$. 
Taking expectations of both sides yields 
\begin{align*}
\beta_4 = T - \expect{Y} -1 + P[Y\geq T-1]
\end{align*}
  Substituting $\expect{Y} = 1-2^{-T}$ and $P[Y\geq T-1] = 2^{-(T-1)}$ yields \eqref{eq:beta4}. 
   
To compute $\beta_3$, assume $T$ is an even positive integer and 
suppose 3-State plays NeverTransmit over $T$ slots.   Let $S$ be the total score and let $Y$ be the number of initial slots in which nobody scores. After 3-State first scores, it will oscillate between being idle and scoring.  Thus
 \begin{align*}
 S &= \frac{T-Y}{2} + \frac{1}{2}1_{\left\{\mbox{$Y$ is odd}\right\}} \\
 \implies \beta_3 &= \frac{T-\expect{Y}}{2} + (1/2)P[\mbox{$Y$ is odd}] 
 \end{align*}
 and using $\expect{Y}=1-2^{-T}$ and $P[\mbox{$Y$ is odd}] = \frac{1-2^{-T}}{3}$ gives 
 the result  \eqref{eq:beta3} when $T$ is even. 
 
 For the case $T$ is odd we have 
$$ S = \frac{T-Y}{2} + \frac{1}{2}1_{\left\{\mbox{$Y$ is even}\right\}} $$
and we use $\expect{Y} = 1-2^{-T}$ and $P[\mbox{$Y$ is even}]=(1/6)(4 - (1/2)^{T-1}) 
$ (which is not the same as the probability that $Y$ is even in the case $T$ is even). 
 \end{proof} 
 
\section{Multi-user MAC} \label{section:multi-user} 

Now consider the MAC problem with $n$ users.  In this version of the problem, we assume the receiver can detect the number of users who transmitted on each slot (as discussed in Section \ref{section:feedback}).   The feedback $F[t]$ given to the users at the end of each slot $t$ is equal to the number of users who transmitted: 
\begin{align*}
F[t] &=0 \implies Idle\\
F[t] &= 1 \implies Success\\
F[t] &\in \{2, \ldots, n\} \implies Collision
\end{align*}
Note that the feedback value $F[t]$ specifies the number of transmitters on slot $t$, but does not indicate \emph{which} users transmitted. 
If $n=2$ then this feedback is equivalent to success/idle/collision feedback.  However, if $n>2$ this feedback is more detailed. It is assumed that all users know the value $n$ at the start, and this value does not change for the timescale of interest.   For knowledge of $n$, one can imagine a ``slot $t=0$'' where all users agree to transmit, so that the feedback signal $F[0]=n$ communicates $n$ to all users.  

The users have no way to distinguish themselves at the start of slot $t=1$, so there are no labels $\{1, 2, \ldots, n\}$ that the users can identify with.  The central goal of this section and the next is to develop an algorithm that all users can independently implement that minimizes the expected time to the first success. This can be viewed as the fundamental learning time required to take a population of $k$ users and distinguish one of them 
through transmission and feedback messaging. Minimizing the time to the first success is useful for either maximizing throughput or for quickly assigning labels to users. 
The next subsection shows how it connects to the finite horizon problem of 
maximizing the Self-Competition score. Section \ref{section:first-capture} presents an (infinite horizon) algorithm for first capturing the channel. 

\subsection{Maximizing throughput} 

Fix $T$ as a positive integer. 
Given $n$ indistinguishable users at the start and the $F[t]$ feedback structure described above, consider the problem of designing an algorithm that is independently implemented by all users that maximizes the total sum number of packet successes over some time horizon $\{1, \ldots, T\}$.  Let $S_1, \ldots, S_n$ be the scores of each user (the users do not know their labels).  Let $V$ be the total score
$$ V = \sum_{i=1}^n S_i$$
Since all algorithms are identical we have 
$$ \expect{S_i} = \frac{\expect{V}}{n} \quad \forall i \in \{1, \ldots, n\}$$
Thus, maximizing $\expect{V}$ is equivalent to maximizing the expected score of each user. 
Let $Y \in \{0, \ldots, T\}$ be the random number of initial slots in which nobody scores ($Y=T$ if nobody ever scores).  Then 
$$ V \leq T-Y \implies \expect{V} \leq T - \expect{Y}$$
This upper bound on $\expect{V}$ can be \emph{achieved} by the policy of letting the first user who scores transmit on the channel interference-free for all future slots (up to and including slot $T$). Thus, maximizing the expected score of each user is equivalent to minimizing $\expect{Y}$. 
While the random variable $Y$ is bounded by the fixed horizon of $T$ slots, so that there is a (small) possibility that no user is ever successful over the $T$ slots, it is convenient to consider the more fundamental problem where there is an infinite horizon, so that $Y$ has no upper bound (this is done in the next subsection).  When the fixed horizon $T$ is large in comparison to the number of users $n$, the  probability that no user is ever successful is so small that the difference between fixed horizon analysis and infinite horizon analysis is insignificant.

\subsection{Minimizing the expected time to the first success} \label{section:first-capture} 

For each positive integer $n$, this subsection develops an algorithm that is independently implemented by $n$ users and seeks to minimize the expected time to the first success.   
There is no deadline, and so this is an infinite horizon problem. 
Let $Z_n$ denote the random time until the first success under the $n$ user algorithm. Define 
$$ z_n= \expect{Z_n} $$
Let $p_n$ denote the transmission probability on the first slot. The idea    is to have the $n$ users independently transmit with probability $p_n$ on slot $1$ and then receive feedback $F[1]$ that specifies the number of transmitters.  If $F[1]=1$ there was a single success and the algorithm terminates. If $F[1] \in \{0, n\}$ then no information that can distinguish the users
is gained and we repeat.  If $F[1]=i$ with $i \in \{2, \ldots, n-1\}$ then the $n$ users are partitioned into a group of size $i$ and a group of size $n-i$, 
the values of $z_i$ and $z_{n-i}$ are compared, the least desirable group is thrown away and the problem is recursively solved on the remaining group with a residual expected time $\min\{z_i, z_{n-i}\}$. 

\subsubsection{Case $n=1$}

If $n=1$ the algorithm is to transmit with probability $p_1=1$, which yields $z_1=1$.  

\subsubsection{Case $n=2$} 

If $n=2$ the algorithm is for both users to independently transmit each slot with probability $p_2=1/2$, so that $z_2 = 2$.

\subsubsection{Case $n=3$} \label{section:k3} 

If $n=3$, the users independently transmit with probability $p \in (0,1)$ on the first slot (the value of $p$ shall be optimized later). 
The feedback after slot $1$ satisfies $F[1] \in \{0, 1, 2, 3\}$.  Based on the value of $F[1]$ do the following: 
\begin{itemize} 
\item $F[1]=0$ (Idle): Repeat. 
\item $F[1]=1$ (Success): Done.
\item $F[1]=2$ (Collision between 2 users): The third user that did not transmit on slot $1$ will transmit alone on slot $2$, while the two users who transmitted on slot $1$ are silent on slot $2$. 
\item $F[1]=3$ (Collision between 3 users): Repeat.
\end{itemize} 
The expected time to the first success is: 
\begin{align}
 z_3 &= \expect{Z_3} \nonumber\\
 &= \sum_{i=0}^3\expect{Z_3|F[1]=i}{3 \choose i} p^i(1-p)^{3-i} \label{eq:sub1} 
 \end{align}
 Under this scheme we have 
 \begin{align*}
   \expect{Z_3|F[1]=0} &= 1+z_3\\
 \expect{Z_3|F[1]=1} &= 1\\
 \expect{Z_3|F[1]=2} &=2\\
  \expect{Z_3|F[1]=3} &= 1 + z_3
 \end{align*}
 Substituting these into \eqref{eq:sub1} gives
 \begin{align*}
&z_3= \\
&\quad (1+z_3)(p^3 + (1-p)^3) + (1)3p(1-p)^2 + (2)3p^2(1-p)
 \end{align*}
 Thus
  $$ z_3 =  \frac{1+3p^2(1-p)}{1-p^3-(1-p)^3}$$
 The value of $p$ that minimizes the right-hand-side is
 $$ p_3  \approx  0.411972 $$
 and so 
 \begin{equation} \label{eq:z3}
 z_3 = \inf_{p \in (0,1)} \left\{    \frac{1+3p^2(1-p)}{1-p^3-(1-p)^3} \right\}  \approx 1.78795
 \end{equation}
 It is surprising that $z_3<z_2$, meaning that it is more efficient to start with $3$ users than to start with 2 users. 

\subsubsection{Case $n\in \{4, 5, 6, \ldots\}$} 

Fix $n$ as the number of users and assume $n \geq 4$.  On the first slot the users independently transmit with probability $p$ (to be optimized later) and react to the feedback $F[1] \in \{0, 1, \ldots, n\}$ as follows: 
\begin{itemize} 
\item $F[1]=0$ (idle): Repeat.
\item $F[1]=1$ (succes): Done. 
\item $F[1]=i\in \{2, \ldots, n-1\}$ (collision of $i$ users): This partitions the users into two groups, one of size $i$ (consisting of those users who have transmitted on the first slot) and one of size $n-i$ (the others).  If $z_i \leq z_{n-i}$ then follow the optimal algorithm for minimizing the time to the first success for $i$ users, utilizing only the $i$ users who transmitted (while the remaining $n-i$ users remain silent).  Else, follow the optimal algorithm for minimizing the time to the first success for $n-i$ users, utilizing only the $n-i$ users who did not transmit in the first slot while the remaining users are silent. 
\item $F[1]=n$ (collision of $n$ users): Repeat. 
\end{itemize} 
The expected time to the first success is then 
\begin{align*}
z_n &= \expect{Z_n} \\
&= \sum_{i=0}^n \expect{Z_n|F[1]=i}{n \choose i} p^i(1-p)^{n-i} \\
&= (1+z_n)(p^n +(1-p)^n) +  np(1-p)^{n-1} \\
& \quad +  \sum_{i=2}^{n-1}(1+ \min\{z_i, z_{n-i}\}){n \choose i} p^i(1-p)^{n-i}
\end{align*}
Thus 
\begin{align*}
z_n &= \frac{p^n + (1-p)^n + np(1-p)^{n-1} }{1-p^n-(1-p)^n}\\
&+ \frac{\sum_{i=2}^{n-1}(1+ \min\{z_i, z_{n-i}\}){n \choose i} p^i(1-p)^{n-i}}{1-p^n-(1-p)^n}\\
&= \frac{1 + \sum_{i=2}^{n-1} \min\{z_i, z_{n-i}\}{n \choose i} p^i(1-p)^{n-i}}{1-p^n-(1-p)^n } 
\end{align*}
The value of $p$ can be chosen to minimize the right-hand-side (call this minimizer $p_n$) and so 
\begin{equation} \label{eq:zk}
z_n = \inf_{p \in (0,1)} \left\{ \frac{1+ \sum_{i=2}^{n-1}\min\{z_i, z_{n-i}\}{n \choose i} p^i(1-p)^{n-i}}{1-p^n-(1-p)^n} \right\} 
\end{equation}
The values of $p_n$ and $z_n$ can be recursively computed in terms of $z_1, \ldots, z_{n-1}$.  The first several values are given below:

\begin{center} 
\begin{tabular}{|c||c|c|}
\hline
$n$ & $p_n$ & $z_n$  \\ \hline \hline
1 & 1 & 1 \\ \hline 
2 &0.5 & 2 \\ \hline 
3 & 0.411972& 1.78795 \\ \hline 
4 & 0.302995 & 2.13454 \\ \hline 
5 &  0.238640& 2.15575 \\ \hline 
6 & 0.191461&2.26246  \\ \hline 
7 &  0.166629& 2.27543 \\ \hline 
\end{tabular} 
\end{center}

\section{Converse} \label{section:converse} 

The $z_n$ and $p_n$ values of the previous section are optimized over algorithms that partition users into two groups and throw away the least desirable group.  It is not clear if gains can be achieved by keeping track of an ever-increasing number of groups 
and having multiple groups probabilistically transmit at the same time.  The information state of the problem is remarkably complex.  For each positive integer $n$, define $z_n^*$ as the infimum expected time for the first success, considering all possible algorithms that can be independently implemented by $n$ users.  Formally, this optimizes over all sequences of Borel measurable functions $\{h_t\}_{t=1}^{\infty}$ with the structure $$h_t:[0,1)\times \{0,1\}^{t-1} \times \{0,1, \ldots, n\}^{t-1}\rightarrow \{0,1\}$$ 
such that the transmission decision for each user $i \in \{1, \ldots, n\}$ is
$$ X_i[t] = h_t(U_i, H_{i,self}[t], (F[1], F[2], \ldots, F[t-1])) \quad \forall t \in \{1, 2, 3, \ldots\}$$ where $\{U_1, \ldots, U_n\}$ are i.i.d. $U[0,1)$ random variables and  $H_{i,self}[t]$ is the history of decisions $(X_i[1], X_i[2], \ldots, X_i[t-1])$ made by the user $i$.  
Note that there is no history before slot 1 so the first function has the structure
$h_1(U_i)$.

Clearly $1\leq z_n^*\leq z_n$, where $z_n$ are the values associated with the proposed algorithm of the previous section.   The author conjectures that $z_n=z_n^*$ for all $n$.  This section proves the conjecture for 
the  special cases $n \in \{1, 2, 3, 4, 6\}$. The cases $n=4$ and $n=6$ are particularly challenging and these are proven by a simple technique of introducing virtual users with enhanced capabilities.

 \begin{thm} For the special cases $n \in \{1, 2, 3, 4, 6\}$ we have $z_n^*=z_n$. 
 \end{thm} 
 
 \begin{proof} 
The fact $z_1^*=z_1=1$ is clear.  The next subsections prove this theorem for the
cases $n \in \{2, 3, 4, 6\}$. 
\end{proof} 


\subsection{Preliminary inequalities} 

Consider $n\geq 2$. Observe that if all $n$ users transmit on the first slot (so that $F[1]=n$) or if none of the users transmit on the first slot (so that $F[1]=0$) then no new information that can distinguish users 
is learned and the infimum remaining time to finish is again $z_n^*$. This observation implies that if $Z$ is the random time to the first success under some particular algorithm, then
\begin{align}
\expect{Z|F[1]=1} &=1 \label{eq:fact-minus1} \\
\expect{Z|F[1]=0} &\geq 1 + z_n^* \label{eq:zfact0}\\
\expect{Z|F[1]=n} &\geq 1 + z_n^* \label{eq:zfactk} 
\end{align} 
where \eqref{eq:fact-minus1} holds because a success on the first slot means $Z=1$; 
\eqref{eq:zfact0}-\eqref{eq:zfactk} hold because if one slot is used and no information is gained, then expected remaining time is again at least $z_n^*$. 

For each positive integer $n \geq 2$ define 
\begin{align*}
\alpha_n &= 1 - (1-\frac{1}{2e})^{1/n}\\
\beta_n &=  (1-\frac{1}{2e})^{1/n}
\end{align*}
For $n\geq 2$ it can be shown that 
$$ 0<\alpha_n\leq \beta_n <1$$

\begin{lem} \label{lem:z-e} For $n\geq 2$ we have $ z_n^* \leq e$.
\end{lem} 
\begin{proof} 
One particular algorithm is to have all $n$ users independently transmit with probability $1/n$, so a success occurs independently with probability $(1-1/n)^{n-1}$ on every slot and the average time to success is $1/(1-1/n)^{n-1}$. Since this is not necessarily the best algorithm, we have
$$ z_n^* \leq \frac{1}{(1-1/n)^{n-1}} \leq e$$
\end{proof}

\begin{lem}  Fix $n\geq 2$.  Fix $\epsilon \in (0,1/2]$. Consider any algorithm that comes within $\epsilon$ of optimality, so that 
$$ z_n^*\leq \expect{Z} \leq z_n^* + \epsilon$$
where $Z$ is the random time to the first success. 
Let $p \in [0,1]$ be the probability of transmitting on the first slot that is used by this algorithm. 
Then 
\begin{equation} \label{eq:p-in} 
p \in [\alpha_n, \beta_n]
\end{equation} 
Furthermore, 
\begin{align}
z_n^*  \geq \frac{1 +\sum_{i=2}^{n-1} \expect{Z-1|F[1]=i}{n \choose i} p^i(1-p)^{n-i} - \epsilon}{1-p^n - (1-p)^n} \label{eq:further}
\end{align}
\end{lem} 
\begin{proof} 
We have 
\begin{align*}
z_n^* + \epsilon &\geq \expect{Z}\\
&=  \sum_{i=0}^{n} \expect{Z|F[1]=i}P[F[1]=i]\\
&\geq (1 + z_n^*)(p^n + (1-p)^n) +   np(1-p)^{n-1} + \sum_{i=2}^{n-1} \expect{Z|F[1]=i}{n \choose i} p^i(1-p)^{n-i}
\end{align*}
where the final inequality holds by \eqref{eq:fact-minus1}-\eqref{eq:zfactk}.  Hence
\begin{equation} \label{eq:to-divide} 
z_n^*(1-p^n-(1-p)^n) + \epsilon \geq p^n + (1-p)^n + np(1-p)^{n-1}+  \sum_{i=2}^{n-1} \expect{Z|F[1]=i}{n \choose i} p^i(1-p)^{n-i}
\end{equation} 
The right-hand-side is at least 1 (since $\expect{Z|F[1]=i}\geq 1$ for all $i$) and so 
\begin{align*}
1 &\leq z_n^*(1-p^n-(1-p)^n) + \epsilon \\
&\leq e(1-p^n - (1-p)^n) + 1/2
\end{align*} 
where the final inequality holds because $z_n^*\leq e$ (by Lemma \ref{lem:z-e}) and the assumption $\epsilon \leq 1/2$. 
Thus
$$ p^n +(1-p)^n \leq 1 - \frac{1}{2e} $$
In particular, both $p$ and $1-p$ are at most $(1-\frac{1}{2e})^{1/n}$. Thus $p  \in [\alpha_n, \beta_n]$, which proves \eqref{eq:p-in}. 
In particular, $1-p^n-(1-p)^n > 0$. Dividing  \eqref{eq:to-divide} by $1-p^n-(1-p)^n$  yields 
\eqref{eq:further}.
\end{proof} 

\subsection{Case $n=2$} 

Suppose there are 2 users.   From \eqref{eq:further} we have 
\begin{align*}
z_2^* &\geq \frac{1-\epsilon}{1-p^2-(1-p)^2} \\
&\overset{(a)}{\geq} (1-\epsilon)\inf_{q \in [\alpha_2, \beta_2]}\left\{\frac{1}{1-q^2-(1-q)^2}\right\}\\
&=2(1-\epsilon)
\end{align*}
where (a) holds because $p \in [\alpha_2, \beta_2]$ (by \eqref{eq:p-in}). 
This holds for all $\epsilon \in (0,1/2]$ and so  $z_2^*\geq 2$. We already know $z_2^*\leq 2$ and so   $z_2^*=2$.

\subsection{Case $n=3$} 

Suppose there are $3$ users.  From \eqref{eq:further} we have 
\begin{align*}
z_3^* &\geq \frac{1 + \expect{Z-1|F[1]=2}3p^2(1-p) -\epsilon}{1-p^3-(1-p)^3}\\
&\overset{(a)}{\geq} \frac{1 + 3p^2(1-p) -\epsilon}{1-p^3-(1-p)^3}\\
&\overset{(b)}{\geq} \inf_{q \in [\alpha_3, \beta_3]} \left\{  \frac{1+3q^2(1-q)-\epsilon}{1-q^3-(1-q)^3} \right\}  
\end{align*}
where (a) holds because if $F[1]=2$ then $Z\geq 2$; (b) holds because $p \in [\alpha_3, \beta_3]$. 
This holds for all $\epsilon \in (0,1/2]$ and so 
\begin{align*}
z_3^* &\geq \inf_{q \in [\alpha_3, \beta_3]} \left\{  \frac{1+3q^2(1-q)}{1-q^3-(1-q)^3} \right\}  \\
&\geq \inf_{q \in (0,1)} \left\{  \frac{1+3q^2(1-q)}{1-q^3-(1-q)^3} \right\}  \\
\end{align*}
The right-hand-side is the definition of $z_3$ in \eqref{eq:z3}.  Thus, $z_3^*=z_3$.

\subsection{Case $n=4$} 

Suppose  there are $4$ users.  From \eqref{eq:further} we have 
\begin{align*}
z_4^* &\geq  \frac{1 + \expect{Z-1|F[1]=2}6p^2(1-p)^2 + \expect{Z-1|F[1]=3}4p^3(1-p) -\epsilon}{1-p^4-(1-p)^4}\\
&\geq  \frac{1 + \expect{Z-1|F[1]=2}6p^2(1-p)^2 + 4p^3(1-p) -\epsilon}{1-p^4-(1-p)^4}
\end{align*}
where the final inequality uses the fact $Z\geq 2$ whenever $F[1]=3$. 
The main challenge is to show the following inequality:
\begin{equation} \label{eq:to-show} 
\expect{Z|F[1]=2} \geq  3
\end{equation} 
The proof of \eqref{eq:to-show} is given in the next subsection.  Substituting \eqref{eq:to-show} into the previous inequality gives: 
\begin{align*}
z_4^* &\geq   \frac{1 + 12p^2(1-p)^2 + 4p^3(1-p) -\epsilon}{1-p^4-(1-p)^4}\\
&\geq \inf_{q \in [\alpha_4, \beta_4]}\left\{    \frac{1 + 12q^2(1-q)^2 + 4q^3(1-q) -\epsilon}{1-q^4-(1-q)^4}\right\} 
\end{align*}
This holds for all $\epsilon \in (0, 1/2]$ and so 
\begin{align*}
z_4^* &\geq \inf_{q \in [\alpha_4, \beta_4]}\left\{    \frac{1 + 12q^2(1-q)^2 + 4q^3(1-q)}{1-q^4-(1-q)^4}\right\} \\
&\overset{(a)}{=} \inf_{q \in [\alpha_4, \beta_4]}\left\{    \frac{1 + \sum_{i=2}^3\min\{z_i, z_{4-i}\}{4 \choose i} q^i(1-q)^{4-i}}{1-q^4-(1-q)^4}\right\} \\
&\geq \inf_{q \in (0, 1)}\left\{    \frac{1 + \sum_{i=2}^3\min\{z_i, z_{4-i}\}{4 \choose i} q^i(1-q)^{4-i}}{1-q^4-(1-q)^4}\right\} 
\end{align*}
where (a) holds because $\min\{z_2, z_2\} = z_2=2$ and $\min\{z_3, z_1\}=z_1=1$.  This lower bound is the same as $z_4$ given in \eqref{eq:zk}. Thus $z_4^*=z_4$. 

\subsection{Proving  \eqref{eq:to-show}} 

It remains to prove \eqref{eq:to-show}:  We have 4 users.  Consider the situation and the end of slot $1$ given that $F[1]=2$.  Let $Z-1$ denote the remaining time until the first success. We shall call this the \emph{2-group situation}: Group A consists of the 2 users who transmitted on the first slot and Group B consists of the 2 users who did not.  We want to show the conditional expected remaining time, given this 2-group situation, is at least 2.  Let $ALG_A$ be an algorithm that the two users in group A independently implement for the remaining slots;  let $ALG_B$ be an algorithm that the two users in group B independently implement for the remaining slots.  Define 
$Time(ALG_A, ALG_B)$ as the expected remaining time (not including the first slot $t=1$) 
to see the first successful packet (from either group).  We want to show that, regardless of the algorithms $ALG_A$ and $ALG_B$, we have 
$$ Time(ALG_A, ALG_B) \geq 2 $$ 

We construct a new system with only two devices, called \emph{virtual devices}, with 
enhanced capabilities.   We show: 
\begin{enumerate} 
\item The new system can emulate the 2-group situation. 
\item Any such emulation in the new system must take at least 2 slots on average. 
\end{enumerate} 
For simplicity we shift the timeline so that the current slot $2$ is now called slot $1$. 
The two virtual devices act independently as follows: On each slot $t \in \{1, 2, 3, \ldots\}$, each virtual device can choose to send any integer number of packets. 
The feedback at the end of slot $t$ is $F[t] \in \{0, 1, 2, \ldots\}$, which is the sum number of packets sent by both virtual devices.   Consider the constraint that both devices must independently implement the same algorithm. Let $R$ be the random time to see the first success, being a slot where exactly one of the devices sends exactly one packet.
Let $r^*$ be the infimum value of $\expect{R}$ over all possible algorithms with this structure. 

\begin{lem} \label{lem:enhanced2} In the situation with two enhanced virtual devices, 
the minimum expected time for the first success is $r^*=2$. 
\end{lem} 

\begin{proof} 
It is clear that $r^*\leq 2$ because one particular algorithm is to have each virtual device independently send either 0 packets or 1 packet every slot, equally likely, and this achieves an expected time of exactly 2.   We now show $r^*\geq 2$. 
Fix $\epsilon$ such that $0<\epsilon\leq 1/2$.  Consider an algorithm that is independently and identically
implemented on both enhanced devices that achieves $\expect{R} \leq r^*+\epsilon$, where $R$ is the random time until the first success.  For 
each $i \in \{0, 1, 2, \ldots\}$, let $a_i$ be the probability that a device sends $i$ packets on the very first slot.  Partition the sample space into the following disjoint events: 
\begin{align*}
A &= \{\mbox{Both algorithms send the same number of packets on the first slot}\}\\
B &= \{\mbox{Success on the first slot}\}\\
C &= (A \cup B)^c
\end{align*}
Then  
\begin{align}
P[A] &= \sum_{i=0}^{\infty} a_i^2\label{eq:padef} \\
P[B] &= 2a_0a_1 \label{eq:pbdef} \\
P[C] &= 1-P[A]-P[B] \nonumber
\end{align}
Observe that 
\begin{align}
\expect{R|A} &\geq 1+r^*\label{eq:main} \\
\expect{R|B} &=1 \label{eq:main2} \\
\expect{R|C}&\geq 2 \label{eq:main3} 
\end{align}
where \eqref{eq:main} holds because the event $A$ uses up one slot and leaves the system in the same state at the end of the first slot (so the remaining expected time is at least $r^*$); \eqref{eq:main2} holds because event $B$ means we are finished in one slot;   \eqref{eq:main3} holds because the event $C$ uses up one slot with no success, and the remaining time is at least one more slot. 
We have 
\begin{align*}
r^*+\epsilon &\geq \expect{R} \\
&= \expect{R|A}P[A] + \expect{R|B}P[B] + \expect{R|C}P[C]\\
&\geq (1+r^*)P[A] + P[B] + 2P[C]\\
&= r^*P[A] + 1 + P[C]
\end{align*}
where the second inequality uses \eqref{eq:main}-\eqref{eq:main3}. 
Thus 
\begin{equation} \label{eq:use-z} 
r^*(1-P[A])   \geq  1 + P[C]-\epsilon
\end{equation} 
Since $r^*\leq 2$ we obtain 
$$ 2(1-P[A])   \geq 1 + P[C] - \epsilon $$
So 
\begin{align*}
1-P[A] &\geq \frac{1+P[C] - \epsilon}{2}\\
&\geq \frac{1+P[C] - 1/2}{2}
\end{align*}
where the final inequality holds because $\epsilon \leq 1/2$. 
Since $P[C]\geq 0$ we obtain 
\begin{equation} \label{eq:one-minus-p}
1-P[A] \geq 1/4
\end{equation} 
Since $1-P[A]>0$ we can rearrange \eqref{eq:use-z}  to yield
\begin{align*}
r^*&\geq \frac{1+P[C] -\epsilon}{1-P[A]} \\
&\overset{(a)}{=} 2 + \frac{P[A]-P[B] - \epsilon}{1-P[A]} \\
&\overset{(b)}{\geq} 2 + \frac{a_0^2+a_1^2- 2a_0a_1-\epsilon}{1-P[A]} \\
&= 2 +  \frac{(a_0-a_1)^2 - \epsilon}{1-P[A]} \\
&\geq 2 - \frac{\epsilon}{1-P[A]}\\
&\overset{(c)}{\geq}2 - \frac{\epsilon}{1/4} 
\end{align*}
where (a) holds because $P[A] + P[B] + P[C]=1$; (b) holds by \eqref{eq:padef}-\eqref{eq:pbdef}; (c) holds by \eqref{eq:one-minus-p}.  Thus
$$ r^*\geq 2 - 4\epsilon $$
This holds for all $\epsilon$ that satisfy $0<\epsilon \leq 1/2$.  Taking a limit as $\epsilon\rightarrow 0$ yields $r^*\geq 2$.
\end{proof} 

We  now show that these two virtual devices can emulate the $k=4$ user scenario with two groups $A$ and $B$. Let $ALG_A$ and $ALG_B$ be two particular algorithms that are independently implemented for (actual) users in Group A and Group B, respectively.  The first enhanced device runs two separate programs: One that emulates an independent user implementing $ALG_A$, the second independently implementing $ALG_B$ as if it is a separate user.  If both $ALG_A$ and $ALG_B$ at this device decide to transmit on the current slot, the device sends two packets.  If only one of $ALG_A$ and $ALG_B$ decide to transmit, the device sends 1 packet.  If neither $ALG_A$ nor $ALG_B$ decide to transmit, the device sends zero packets.  The second device does a similar emulation independently.  The feedback signaling $F[t]$ on each slot $t$ is the same as if there were 4 users that were initially configured in the 2-group scenario.  Hence, the expected time to achieve the first success is the same as $Time(ALG_A, ALG_B)$ and, since this is also a situation where two enhanced devices independently implement the same algorithm, we have 
$$ Time(ALG_A, ALG_B) \geq r^*=2$$
where $r^*=2$ follows from the previous lemma. This proves \eqref{eq:to-show}. 

\subsection{The case $n=6$} 

Suppose there are 6 users. From \eqref{eq:further} we have 
\begin{align}
z_6^*  \geq \frac{1 +\sum_{i=2}^{5} \expect{Z-1|F[1]=i}{6 \choose i} p^i(1-p)^{6-i} - \epsilon}{1-p^6 - (1-p)^6} \label{eq:further2}
\end{align}
We now claim: 
\begin{align}
\expect{Z-1|F[1]=2} &\geq z_2 = \min\{z_2, z_4\} \label{eq:six3}\\
\expect{Z-1|F[1]=3} &\geq z_3 = \min\{z_3, z_3\} \label{eq:six4}\\
\expect{Z-1|F[1]=4} &\geq z_2 = \min\{z_2, z_4\} \label{eq:six5}
\end{align}
The claims \eqref{eq:six3}-\eqref{eq:six5} are nontrivial and shall be proven later. 
Utilizing these in the previous inequality gives
\begin{align*}
z_6^*&\geq \frac{1  + \sum_{i=2}^5 \min\{z_i, z_{6-i}\}{6 \choose i}p^i(1-p)^{6-i} - \epsilon}{1-p^6 - (1-p)^6}\\
&\geq \inf_{q \in [\alpha_6, \beta_6]}\left\{ \frac{1  + \sum_{i=2}^5 \min\{z_i, z_{6-i}\}{6 \choose i}q^i(1-q)^{6-i} - \epsilon}{1-q^6 - (1-q)^6}\right\}
\end{align*}
This holds for all $\epsilon \in (0, 1/2]$ and so 
\begin{align*}
z_6^*&\geq \inf_{q \in [\alpha_6, \beta_6]}\left\{ \frac{1  + \sum_{i=2}^5 \min\{z_i, z_{6-i}\}{6 \choose i}q^i(1-q)^{6-i} }{1-q^6 - (1-q)^6}\right\}\\
&\geq  \inf_{q \in (0,1)}\left\{ \frac{1  + \sum_{i=2}^5 \min\{z_i, z_{6-i}\}{6 \choose i}q^i(1-q)^{6-i} }{1-q^6 - (1-q)^6}\right\}\\
\end{align*}
This lower bound is the same as $z_6$ by \eqref{eq:zk}. Thus, $z_6^*=z_6$. 

\subsection{Proving \eqref{eq:six3} and \eqref{eq:six5}} 

The given information $F[1]=2$ creates a system of two groups $A$ and $B$, where group $A$ has 2 users, group $B$ has 4 users, the users know their group and know the total number of users in their group, but cannot distinguish themselves from others in their group.  The situation $F[1]=4$ is identical and so it suffices to prove \eqref{eq:six3}.  In this situation, it suffices to prove that the expected remaining time is at least $z_2^*=2$. 
This situation can be emulated by a 2-enhanced device system where the first device emulates the first user in group $A$ and two users in group $B$, while the second device emulates the second user in group $A$ and the other two users in group $B$.  We already know that any such emulation requires at least $2$ slots (Lemma \ref{lem:enhanced2}), 
and so the expected remaining time in the 6-user system, given $F[1]=2$, is at least 2.  Since $z_2=2$, this proves \eqref{eq:six3}. 

\subsection{Proving \eqref{eq:six4}} 

Suppose there are 6 users and $F[1]=3$. This information creates two groups at the start of the second slot,  Group $A$ and Group $B$, each with three users.  Let $r^*$ represent the infimum expected remaining time to see the first success in this situation.   
Clearly $r^*\leq z_3$ because we can simply throw away Group B and implement the optimal algorithm for 3 users.  We want to show the reverse inequality $r^*\geq z_3$.  

This situation can be emulated by 3 virtual devices with the enhanced capability of sending either 0, 1, or 2 packets on every slot.   Let $y^*$ represent the infimum expected time to see the first success when these 3 virtual users independently implement the same algorithm (the infimum is taken over all possible algorithms).  First note that $y^*\leq r^*$
because the three virtual devices can emulate the two groups of actual users as follows: Consider any algorithms $ALG_A$ and $ALG_B$ that are independently used by (actual) users in Groups A and B, respectively.  Similar to the $k=4$ situation, 
each virtual device can simulate one actual user implementing $ALG_A$  and 
one actual user independently implementing $ALG_B$. As before, on any slot when both the  $ALG_A$ and $ALG_B$ user transmit, that virtual device sends 2 packets; on any slot when $ALG_A$ sends but $ALG_B$ does not, or vice versa, the device sends 1 packet; on any slot when neither $ALG_A$ nor $ALG_B$ sends, the device sends 0 packets. The time to first success in this emulation is the same as the time to first success in the actual situation, and so $y^*\leq r^*$. 

It remains to show $y^*\geq z_3$ (which implies the desired inequality $r^*\geq z_3$). 
Clearly $1\leq y^*\leq z_3\leq 2$. 
Fix $\epsilon \in (0,1/4)$. Consider an algorithm that is independently 
implemented on the system of 3 virtual users with enhanced capabilities that yields 
\begin{equation} \label{eq:y-approx}
 y^* \leq \expect{Z}\leq y^* + \epsilon
 \end{equation} 
where $Z$ is the random time to the first success.   Let $a,b,c$ be the probabilities that a single device sends 0, 1, or 2 packets, respectively, on the first slot.  Partition the sample space into the following disjoint events  $A, B, C$: 
\begin{align*}
A &= \{\mbox{All 3 users send the same number of packets}\}\\
B &= \{\mbox{Success on first slot}\} \\
C &= (A \cup B)^c 
\end{align*}
Then 
\begin{align*}
&P[A] = a^3+b^3 + c^3\\
&P[B] = 3ba^2\\
&P[A]+P[B] + P[C] = 1
\end{align*}
Furthermore
\begin{align*}
\expect{Z|A} &\geq 1+y^*\\
\expect{Z|B}&=1\\
\expect{Z|C}&\geq 2
\end{align*}
Thus, starting with \eqref{eq:y-approx} gives 
\begin{align}
y^* + \epsilon &\geq \expect{Z} \nonumber \\
&= \expect{Z|A}P[A]  + \expect{Z|B}P[B] + \expect{Z|C}P[C] \nonumber\\
&\geq (1+y^*)P[A] + P[B]+ 2P[C] \label{eq:starbucks}
\end{align}
In particular 
$$ (1+y^*)P[A] \leq y^* + \epsilon \leq y^* + 1/4$$
and so 
$$ P[A] \leq \frac{y^*+1/4}{1+y^*} \leq \frac{2 + 1/4}{1+2} = 3/4$$
where we have used the fact that $(y+1/4)/(1+y)$ is increasing on the domain $y>0$ and $y^*\leq 2$. Thus $P[A] \leq 3/4$. 

Also from \eqref{eq:starbucks}:
\begin{align*}
y^*&\geq \frac{P[A] + P[B] + 2P[C]-\epsilon}{1-P[A]} \\
&= 1 + \frac{1-P[B]-\epsilon}{1 -P[A]} \\
&\overset{(a)}{\geq} 1 + \frac{1-P[B]}{1-P[A]} - 4\epsilon \\
&=1 + \frac{1-3ba^2}{1-a^3-b^3-c^3} - 4\epsilon\\
&= 1 + \frac{1-3(1-a-c)a^2}{1-a^3-(1-a-c)^3-c^3} - 4\epsilon\\
&\geq 1 + \inf_{(a,c) \in \Lambda}\left\{\frac{1-3(1-a-c)a^2}{1-a^3-(1-a-c)^3-c^3} \right\} - 4\epsilon
\end{align*}
where  (a) holds because $1-P[A] \geq 1/4$ and where $\Lambda$ is defined 
$$ \Lambda = \{(a,c) \in \mathbb{R}^2 : a\geq 0, c\geq 0, a+c\leq 1, a^3+(1-a-c)^3+c^3\leq 3/4\}$$
This holds for all $\epsilon \in (0,1/4)$. Taking a limit as $\epsilon \rightarrow 0$ gives 
$$ y^* \geq 1 + \inf_{(a,c) \in \Lambda}\left\{\frac{1-3(1-a-c)a^2}{1-a^3-(1-a-c)^3-c^3} \right\}$$
It can be shown that  
 the infimum of this 2-variable problem is achieved at a point $(a^*, c^*) \in \Lambda$ such 
 that $0.5\leq a^*\leq 1$.   Define the function $f:[0, 1-a^*]\rightarrow \mathbb{R}$ by 
 $$ f(c) = 1 - (a^*)^3 - (1-a^*-c)^3 - c^3$$
 By concavity of $f$ we have 
 $$ f(c^*) \leq f(0) + f'(0)c^*$$
 and so 
 \begin{equation} \label{eq:uses}
  f(c^*)\leq 1 - (a^*)^3 - (1 - a^*)^3 + 3(1-a^*)^2c^*
  \end{equation} 
  and so 
 \begin{align}
 y^* &\geq 1 + \frac{1-3(1-a^*-c^*)(a^*)^2}{1-(a^*)^3-(1-a^*-c^*)^3-(c^*)^3} \nonumber \\
 &\geq 1 + \frac{1-3(1-a^*-c^*)(a^*)^2}{1 - (a^*)^3 - (1 - a^*)^3 + 3(1-a^*)^2c^*} \label{eq:chari} 
 \end{align}
 where the final inequality uses \eqref{eq:uses} and the fact that the numerator is nonnegative. 
 For a fixed $a^* \in [0.5, 1]$, the right-hand-side of \eqref{eq:chari} is a nondecreasing function
 of $c^*$ and so it is greater than or equal to its corresponding value when $c^*$ is replaced by $0$: 
 \begin{align*}
 y^*&\geq 1 + \frac{1-3(1-a^*)(a^*)^2}{1-(a^*)^3 - (1-a^*)^3}\\
 &\geq 1 + \inf_{a \in (0,1)} \left\{\frac{1 - 3(1-a)a^2}{1 - a^3 - (1-a)^3 }\right\} \\
&= \inf_{a \in (0,1)} \left\{\frac{1 + 3a(1-a)^2}{1 - a^3 - (1-a)^3 }\right\} = z_3
\end{align*}

\section{Multiple channels} \label{section:extension} 


This section briefly explores the problem of minimizing the expected time to first capture when there are $n$ users and $m$ orthogonal channels.  Time is slotted and each channel can support one packet transmission per slot.  
Each user can transmit on multiple channels simultaneously. 
On each slot, each user chooses a subset of the $m$ channels over which to transmit.   The feedback at the end of each slot $t$ is the vector $F[t]=(F_1[t], \ldots, F_m[t])$ where $F_i[t] \in \{0,1,2, \ldots, n\}$ represents the number of transmitters on channel $i$.  If $F_i[t]=0$ then channel $i$ was idle; if $F_i[t]=1$ then channel $i$ had a success; if $F_i[t]=j$ for some $j \in \{2, \ldots, n\}$ 
then channel $i$ experienced a collision of $j$ packets.   Let $Z \in \{1, 2, 3, \ldots\}$ be the random number of slots until at least one channel has a success.  The goal is to minimize $\expect{Z}$. 
Is it optimal to use all channels separately and independently?   The  answer is different when $n=2$ versus $n=3$.

\subsection{Two users} 

Suppose there are two users and $m$ orthogonal channels.  One transmission
strategy is to have each user transmit with probability $1/2$ independently over each channel and each slot. Under this strategy, the probability of having at least one success on a given slot is $1-(1/2)^m$ and so $Z$ is geometrically distributed with
$$ \expect{Z} = \frac{1}{1-(1/2)^m}$$

We now show this strategy is optimal under the assumption that both users are indistinguishable and independently implement identical policies.  Let $z^*$ denote the infimum expected capture time over all such policies. Fix $\epsilon\in (0, 1)$. Consider a policy that is independently implemented by both users that achieves an expected capture time within $\epsilon$ of the optimal $z^*$.  Let $Z$ denote the random time to first success under this policy, so
\begin{equation} \label{eq:baz1} 
\expect{Z}\leq z^*+\epsilon
\end{equation} 
 Enumerate the $2^m$ subsets of the $m$ channels and let $p_i$ denote the probability that a user chooses subset $i \in \{1, \ldots, 2^m\}$ for transmission on the first slot (under this particular policy). Observe that slot 1 has at least one success if and only if both users choose different subsets, and so
$$ \theta = P[\mbox{Success in slot 1}] = 1 - \sum_{i=1}^{2^m}p_i^2$$
On the other hand, if both users choose identical subsets on slot 1 then no information that can distinguish the users is gained and the expected remaining time as at least $z^*$. 
Thus
\begin{equation} \label{eq:baz2}
\expect{Z} \geq 1 + (1-\theta)z^*
\end{equation} 
Combining \eqref{eq:baz1} and \eqref{eq:baz2} gives
$$z^* +\epsilon\geq 1 + (1-\theta)z^*$$
and so 
\begin{equation} \label{eq:and-so}
z^* \geq \frac{1-\epsilon}{\theta} \geq \frac{1-\epsilon}{\theta^*}
\end{equation} 
where $\theta^*$ is the optimal objective value for the following convex optimization problem
\begin{align*}
\mbox{Maximize:} \quad1- \sum_{i=1}^{2^m}q_i^2 \\
\mbox{Subject to:} \quad \sum_{i=1}^{2^m}q_i=1
\end{align*}
Using a Lagrange multiplier argument, it is not difficult to show the solution to this problem is a uniform distribution, so  $q_i^*= \frac{1}{2^m}$ for all $i \in \{1, \ldots, 2^m\}$. Hence,  $\theta^*= 1-(1/2)^m$. Substituting this value of $\theta^*$ into \eqref{eq:and-so} gives
$$ z^* \geq \frac{1-\epsilon}{1-(1/2)^m}$$
This holds for all $\epsilon \in (0,1)$ and so 
$z^* \geq \frac{1}{1-(1/2)^m}$.  

\subsection{Three users}

Suppose there are three users and $m$ channels.  Each user independently implements the same policy for choosing subsets of channels over which to transmit on each slot.  For a particular policy, let $\beta$ be the probability that all three users choose the same subset for transmission on the first slot. Let $\theta$ be the probability that, on the first slot, no channel has exactly one transmitter  but at least one channel has exactly 2 transmitters.  Let $A$ be the set of all $(\beta, \theta)$ vectors that can be achieved (considering all possible policies).  
Define $z^*$ as the infimum expected time to first capture. 
Using an argument similar to the single channel case treated in Section \ref{section:k3}, it can be shown that 
$$ z^* = \inf_{(\beta, \theta) \in A} \left\{ \frac{1+\theta}{1-\beta}\right\}$$

It is not trivial to optimize over all $(\beta, \theta)$ values in the set $A$.  For simplicity, consider the special case with just $m=2$ channels. Let $p$ be the probability that each user transmits on channel 1 on the first slot.  Let $q$ be the conditional 
probability of transmitting on the second channel on the first slot, given the user transmitted on the first channel.  Let $r$ be the conditional probability of transmitting on the second channel on the first slot, given the user did \emph{not} transmit on the first channel.  Then
\begin{align*}
\beta &= p^3(q^3 + (1-q)^3) + (1-p)^3(r^3 + (1-r)^3)\\
\theta &= p^3(3q^2(1-q)) + (1-p)^3(3r^2(1-r)) + 3p^2(1-p)(1  - 2q(1-q)(1-r) - r(1-q)^2)
\end{align*}
Using Matlab to minimize $(1+\theta)/(1-\beta)$ over $p,q, r \in [0,1]$ gives
$$ p^*=0.5, q^*=0, r^*=1, z^* = 1.3333$$
This policy does not treat each channel separately and independently: On the first slot it is optimal to transmit over \emph{either} channel 1 or channel 2, but not both. Thus, unlike the case with only two users, it is not desirable to have a single user transmit over both channels simultaneously. 

It is instructive to compare this to the policy that  independently implements the min capture time policy   of Section \ref{section:first-capture} separately over both of the two channels.  More generally, one can minimize $(1+\theta)/(1-\beta)$ over all $(\beta, \theta)$ values achievable by policies that independently use the channels with the same transmit probability $p$.  Under this more restrictive class of policies we have 
\begin{align*}
\beta &= (p^3+(1-p)^3)^2\\
\theta &= (p^3 + (1-p)^3)(3p^2(1-p)) + 3p^2(1-p)(1 - 3p(1-p)^2) 
\end{align*}
Numerically minimizing $(1+\theta)/(1-\beta)$ over $p \in [0,1]$ gives 
$$ \tilde{p} = 0.360882 , \tilde{z} = 1.34373$$
which is strictly worse than $z^*=1.3333$. 

\section{Conclusion} 

A MAC game was introduced. Unlike related prisoner dilemma games where Tit-for-Tat policies tend to win, it was shown that a 4-State policy that has a randomized initial phase consistently won each competition.  The policy 4-State (and a closely related policy 3-State) was shown to maximize the expected number of points when competing against an independent version of itself.   A closely related problem of minimizing the expected time required to first capture the channel was explored and an optimal algorithm was developed for the special case when the number of users $n$ is in the set $\{1, 2, 3, 4, 6\}$.  The optimality proof uses a technique that introduces virtual users with enhanced capabilities. An efficient policy was developed for any positive integer $n$.  Whether or not this policy is optimal for all $n$  is left as an open question.  Finally, the complexities of the multi-channel case were explored.  It was shown that, for the context of minimizing the expected time to first capture, it is optimal to independently use all channels when the number of users is only two.  However, a counter-example was shown for the case of three users and two channels: Optimality requires each user to correlate its transmission decisions over each channel, rather than independently use each channel.

\bibliographystyle{unsrt}
\bibliography{../../../../../../latex-mit/bibliography/refs}
\end{document}